\documentclass[12pt]{amsart}

\usepackage{graphicx,color}
\usepackage{setspace}
\usepackage{amsmath, amssymb,amsthm, amsfonts}
\usepackage{accents}

\newtheorem{thm}{Theorem}[section]
\newtheorem{lma}{Lemma}[section]

\newtheorem{cor}{Corollary}[section]

\theoremstyle{definition}

\theoremstyle{remark}

\numberwithin{equation}{section}

\newcommand{\R}{\mathbb R}

\newcommand{\be}{\begin{equation}}
\newcommand{\ee}{\end{equation}}
\newcommand{\bee}{\begin{equation*}}
\newcommand{\eee}{\end{equation*}}

\def\p{\partial}

\def\la{\langle}
\def\ra{\rangle}
\def\lf{\left}
\def\ri{\right}
\def\Pi{\displaystyle{\mathbb{II}}}

\def\e{\epsilon}

\def\a{\alpha}

\def\wt{\widetilde }

\def\n{\scriptscriptstyle{\mathbf{N}}}

\begin{document}
\title[]
{Quasilocal energy in Kerr spacetime}

\author{Jian-Liang Liu}
\address[Jian-Liang Liu]{Department of Mathematics, Shantou University, Guangdong 515063, P.R.China.}
 \email{liujl\_mc@163.com}

\author{Luen-Fai Tam{$^1$}}
\address[Luen-Fai Tam]{The Institute of Mathematical Sciences and Department of
 Mathematics, The Chinese University of Hong Kong, Shatin, Hong Kong, China.}
 \email{lftam@math.cuhk.edu.hk}
\thanks{$^1$Research partially supported by Hong Kong RGC General Research Fund \#CUHK 14305114}

\renewcommand{\subjclassname}{
  \textup{2010} Mathematics Subject Classification}

\date{February, 2016}

\begin{abstract} In this work we study the quasilocal energy as in \cite{SCLN:2013} for a constant radius surface in Kerr spacetime in Boyer-Lindquist coordinates. We show that under suitable conditions for isometric embedding, for a stationary observer the quasilocal energy defined in \cite{SCLN:2013}  for constant radius in a Kerr like spacetime is exactly equal to the Brown-York quasilocal energy \cite{Brown:1992br}. By some careful estimations, we show that for a constant radius surface in the Kerr spacetime which is outside the   ergosphere the embedding conditions for the previous result are satisfied. We discuss extremal solutions as described in \cite{JLTH}. We prove a uniqueness result. We find all extremal solutions in the Minkowski spacetime. Finally, we show that near the horizon of the Kerr spacetime for the small rotation case the extremal solutions are trivial.

\end{abstract}

\maketitle

\markboth{Jianliang Liu and Luen-Fai Tam}
{Quasilocal energy in Kerr spacetime}


\section{Introduction}

In this work, we want to discuss the quasi-local energy (QLE) as in \cite{SCLN:2013,SCLN:2015,SLN2015,JLTH} for some spacelike surfaces on a Kerr-like spacetime. Let us recall the formulation of such a
QLE.
From the covariant Hamiltonian formalism \cite{Chen:1998aw,Chen:2005hwa}, the conserved current is defined by the Hamiltonian 3-form $\mathcal{H}(\mathbf{N})$ (on the space-like hypersurface) under the infinitesimal diffeomorphism generated by $\mathbf{N}$. The conserved quantity is then the integration over the space-like region $\Omega$ which reduces to the boundary integration of the total derivative term in $\mathcal{H}(\mathbf{N})$ \emph{on shell}. The field equations are preserved under any modification of the total derivative term, which changes the boundary term and then changes the   value of the conserved quantity. Modifying the boundary term implies a different boundary condition (corresponding to a different pseudotensor expression) and the choice of reference (corresponding to the frame choice of the pseudotensor). For a specific boundary expression, the choice of reference is still arbitrary. The difficulty comes from choosing a reasonable reference. There are several different strategies for choosing the reference. We will follow the $4D$ isometric matching as in \cite{Wu:2011wk,Wu:2012mi,SCLN:2013,SCLN:2015,SLN2015}. The idea of the $4D$ isometric matching $g_{\mu\nu}\doteq\bar{g}_{\mu\nu}$  is that the metric $\bar{g}_{\mu\nu}$ of the background spacetime matches the physical metric $g_{\mu\nu}$ on the quasi-local 2-boundary. One can imagine that there is an observer in a specific spacetime who measures the conserved quantity of the physical world. The matching of the metric may be regarded as the calibration of the measurement on the 2-surface. Consequently, the conserved quantity is obtained by the specific displacement vector field $\mathbf{N}$, e.g.\ energy conservation for a time-like displacement; linear momentum for a space-like transition; angular momentum for a rotation and the center-of-mass moment for a boost displacement.

We will  focus on the quasi-local \emph{energy} of the Kerr spacetime. The background spacetime is chosen to be the Minkowski spacetime and $\mathbf{N}$ to be the time-like Killing vector of the background. The reference choice appears as the Jacobian of the background coordinate system, which also determines the displacement.

The $4D$ matching $g_{\mu\nu}\doteq\bar{g}_{\mu\nu}$ gives 10 constraints on the 12 independent unknowns of the Jacobian, and it reduces to two freedoms: one corresponds to the 2-surface isometric embedding (we will call it ``the embedding freedom'') and the other to the displacement, or in other words, the observer dependence (we will call it ``the boost freedom''). The 2-surface isometric embedding is unique up to one free function, which is called ``the admissible $\tau$'' in \cite{WaYa2008}. The other freedom is from the remaining 7 constraints of the 8 unknowns. The Hamiltonian boundary expression is then a functional of these two free choices. There is no unique value of energy because different observers have different measurements. One could find the critical value via the variation with respect to the free choices \cite{JLTH,SLN2015}.

Consider the physical spacetime $(M,\mathbf{g})$ so that the metric is a Kerr-like metric:

\be\label{e-Kerr-like-1}
 \mathbf{g}=Fdt^2+2Gdtd\varphi+Hd\varphi^2+R^2dr^2+\Sigma^2d\theta^2
\ee
where $F, G, H, R, \Sigma$ are functions of $r,\theta$ only. The background we choose is the Minkowski spacetime $(\bar{M},\bar{\mathbf{g}})$:
\be
\bar{\mathbf{g}}= -dT^2+dX^2+dY^2+dZ^2.
\ee
 We use the Chen-Nester-Tung (CNT) quasi-local expression in \cite[eq. (4)]{SCLN:2015}, which is
\be
\mathcal{B}(\mathbf{N})=\frac{1}{2\kappa}(\Delta\Gamma^{\alpha}{}_{\beta}\wedge\iota_{\n}\eta_{\alpha}{}^{\beta}+\bar{D}_{\beta}N^{\alpha}\Delta\eta_{\alpha}{}^{\beta}),
\ee
where $\kappa=8\pi$, $\iota_{\n}\eta_{\alpha}{}^{\beta}=\sqrt{-g}N^{\mu}g^{\beta\gamma}\epsilon_{\alpha\gamma\mu\nu}dx^{\nu}$ and $\Delta \alpha:=\alpha-\bar{\alpha}$ is the difference of the physical field and the reference one. Note that the second term includes $\Delta\eta_{\alpha}{}^{\beta}:=\frac{1}{2}(\sqrt{-g}g^{\beta\gamma}-\sqrt{-\bar{g}}\bar{g}^{\beta\gamma})\epsilon_{\alpha\gamma\mu\nu}dx^{\mu}\wedge dx^{\nu}$, which vanishes for the   $4D$ matching condition $g_{\mu\nu}\doteq\bar{g}_{\mu\nu}$. The quasi-local energy can then be simplified to
\be
E(\mathbf{N},\Omega)=\oint_{\partial\Omega}\mathcal{B}(\mathbf{N})=\oint_{\mathcal{S}}\frac{1}{2\kappa}\Delta\Gamma^{\alpha}{}_{\beta}\wedge\iota_{\n}\eta_{\alpha}{}^{\beta}.
\ee
In this work, we will consider the spacelike two surfaces $\mathcal{S}(r_0)$  in $M$ given by $r=r_0$, $t=t_0$ where $r_0, t_0$ are constants. We consider   embeddings in $\bar  {M}$ of the form:
\be\label{e-TXYZ}
T=T(t,r,\theta), X=\rho\cos(\varphi+\Phi),  Y=\rho\sin(\varphi+\Phi), Z=Z(t,r,\theta)
\ee
where $\rho$ and $\Phi$ are functions of $(t,r,\theta)$. As mentioned above, the $4D$ matching has two free choices, which will be chosen to be $x=T_r, y=T_\theta$. We choose $\mathbf{N}=\p_T$. The next step is to find the critical value by varying $x, y$. One then obtains (\cite{SCLN:2013} (50) and (51))
\be\label{e-xy-system-1}
\left\{
  \begin{array}{ll}
    y_\theta=  &\displaystyle{ -\frac{(\Sigma^2H)_r}{2HR^2}x-\frac{\Sigma H_\theta-2H\Sigma_\theta}{2H\Sigma}y} \\
    x_\theta= & \displaystyle{\frac{R_\theta}{R}x+\lf(\frac{(\Sigma^2H)_r}{2H\Sigma^2}-\frac{\a\beta+xy H_\theta}{2H\ell}\ri)y}.
  \end{array}
\right.
\ee
Substitute these equations into the QLE expression (\cite{SCLN:2013}, (48)), then it becomes (see Appendix)
\be\label{e-QLE-Kerrlike}
E(\mathbf{N},\Omega(r);x,y)=-\int_0^{2\pi}\lf(\int_0^\pi\lf(\frac{AH_\theta}{D}\frac{xy}{\beta \ell}+\frac AD\frac\a\ell+\frac CD\frac1\beta\ri)d\theta\ri)d\varphi
\ee
where
\be\label{e-ABC-1}
\left\{
  \begin{array}{ll}
    A=&\Sigma(H\Sigma^2)_r,\\
C=&2HR^2\lf(H_{\theta\theta}\Sigma-H_\theta\Sigma_\theta-2\Sigma^3\ri)\\
D=&16\pi H^\frac12 R^2     \Sigma  \\
\a=&\lf(x^2\Sigma^2   + R^2(y^2 + \Sigma^2)\ri)^\frac12\\
\beta=&\lf(-H_{\theta}^2 + 4H(y^2 + \Sigma^2)\ri)^\frac12\\
\ell=&y^2+\Sigma^2.
  \end{array}
\right.
\ee
Here $\Omega(r)$ is the domain in the time slice $t=t_0$ with  boundary $\mathcal{S}(r)$ and $E$ depends on the choice of $x, y$.

It is obvious that \eqref{e-xy-system-1} has the trivial solution   $x\equiv0$, $y\equiv0$, which may be considered to correspond to a stationary observer. Our first result is:

{\it Suppose $\mathcal{S}(r)$ has positive Gaussian curvature and can  be isometrically embedded in $\R^3$ as a surface of revolution, then $E(\mathbf{N},\Omega(r);0,0)$ is exactly equal to the Brown-York QLE in \cite{Brown:1992br}}.

 It is easy to see that for the Kerr metric, in the slow rotation case, $\mathcal{S}(r)$ satisfies the above embedding condition. In \cite{SCLN:2013}, using the slow rotation approximation of Martinez \cite{Martinez}, it was proved the $E(\mathbf{N},\Omega(r);0,0)$ is equal to the Brown-York QLE in the slow rotation approximation. It was also proved that $E(\mathbf{N},\Omega(r);0,0)$ is the Brown-York QLE for the Schwarzschild metric. Our result says that in fact for the slow rotation case in the Kerr spacetime, they are {\it exactly} the same.

One question is whether this is still true for a general rotation. Direct computations show  that  in the extremal case for the Kerr metric, the horizon has negative Gaussian curvature somewhere. In fact,  it is known  \cite{Smarr} that the Kerr
horizon cannot be embedded globally in $\R^3$ whenever
$a > \sqrt3 m/2$. Here $a$ is the angular momentum per unit mass and $m$ is the mass. We always assume that $a\le m$.  Physically one would like to consider $\mathcal{S}(r)$ which is   outside the ergosphere. Our second result is:

{\it For the Kerr spacetime, if $r\geq2m$, then $\mathcal{S}(r)$ has positive Gaussian curvature and can be isometrically embedded in $\R^3$ as a surface of revolution.}

Hence  our first result applies to this situation. In particular, we obtain an explicit formula for the Brown-York quasi-local energy for $\mathcal{S}(r)$.

Our next result is to consider solutions of \eqref{e-xy-system-1} for Kerr spacetime. For Minkowski spacetime, i.e. $a=m=0$, then one can find nontrivial solutions to \eqref{e-xy-system-1}, see \cite{JLTH}. In fact, one can find all the solutions in the Minkowski spacetime. This follows from an uniqueness result. Namely, we prove:

{\it For the Kerr spacetime, two sets of solutions to \eqref{e-xy-system-1} with bounded derivatives are equal if they are equal at $\theta=0$ or at $\theta=\pi$. In particular, the Minkowski spacetime, the solutions are of the form $x=-k\cos\theta$, $y=kr\sin\theta$, where $k$ is a constant.}

  Hence for the Minkowski spacetime, $E(\mathbf{N},\Omega(r);x,y)=0$ for all solutions $x, y$ to \eqref{e-xy-system-1}.

On the other hand, if  there is a horizon, then we may only have trivial solutions. Let $r_+=m+(m^2-a^2)^\frac12$. We prove that:

{\it If $0\le a \ll m$ and if $r>r_+$ which is close enough to $r_+$, then any solution $x, y$ to \eqref{e-xy-system-1} with $|x_\theta|, |y_\theta|$ being bounded must be trivial, i.e.  $x\equiv0$, $y\equiv0$.}

The organization of the paper is as follows: In section \ref{s-QLE}, we will prove that under suitable conditions $E(\mathbf{N},\Omega(r);0,0)$ is equal to the Brown-York QLE. In section \ref{s-embedding}, we will discuss the embedding problem for surfaces $\mathcal{S}(r)$ in the Kerr spacetime. In section \ref{s-solution}, we will discuss the solutions to the system \eqref{e-xy-system-1}.

{\it Acknowledgement}: The  authors would like to thank Prof. James M. Nester for   helpful discussions. The first author would like to thank the visiting support by the School of Mathematical Sciences, Capital Normal University and Morningside Center of Mathematics, Chinese Academy of Sciences.

 \section{CNT QLE and Brown-York QLE}\label{s-QLE}

 Let $(M,\mathbf{g})$ be a Kerr like spacetime with $\mathbf{g}$ given as in \eqref{e-Kerr-like-1}, with $0\le \theta\le \pi$, $0\le\varphi<2\pi$.
  We assume that the slice $\hat M= \{t=t_0\}$ is spacelike, where $t_0$ is a constant. Then the induced metric on $\hat M$ is:
$$
\hat g=H^2d\varphi^2+R^2dr^2+\Sigma^2d\theta^2.
$$
Let $\mathcal{S}(r_0) $ be the surface $r=r_0$ in $\hat M$, where $r_0$ is a constant. The induced metric on $\mathcal{S}(r_0)$ is given by
$$
d\sigma^2=\Sigma^2d\theta^2+Hd\varphi^2.
$$
We assume this is a closed surface.
\begin{lma}\label{l-BY-1} Let $\kappa$ be the mean curvature of $\mathcal{S}(r_0)$ with respect to the unit normal $\nu=R^{-1}\p_r$. Then
\bee
\frac1{8\pi}\int_{\mathcal{S}(r_0)}\kappa d\sigma= \int_0^{2\pi}\lf(\int_0^\pi \frac{A\a}{D\ell}d\theta\ri)d\varphi
\eee
where $A, \a, D, \ell$ are as in \eqref{e-ABC-1} with $\a$ being evaluated at $x=0, y=0$, and $d\sigma$ is the area element on $\mathcal{S}(r_0)$.
\end{lma}
\begin{proof}
Note that $H, R, \Sigma$ are independent of $\varphi$.
Let $\nabla$ be the covariant derivative of the slice. Then the mean curvature is
\be\label{l-mean-1}
\begin{split}
\kappa=&-\lf(H^{-1}\la \nabla_{\p_\varphi}\p_\varphi,\nu\ra+\Sigma^{-2}\la \nabla_{\p_\theta}\p_\theta,\nu\ra\ri).
\end{split}
\ee
On the other hand,
\bee
\begin{split}
\la \nabla_{\p_\varphi}\p_\varphi,\nu\ra=&\frac1R\la \nabla_{\p_\varphi}\p_\varphi,\p_r\ra\\
=&-\frac1R\la\p_\varphi, \nabla_{\p_\varphi}\p_r\ra\\
=&-\frac1R\la\p_\varphi, \nabla_{\p_r} \p_\varphi \ra\\
=&-\frac1{2R}\p_r\la\p_\varphi,\p_\varphi\ra\\
=&-\frac{H_r}{2R}.
\end{split}
\eee
 Similarly,
 \bee
 \la \nabla_{\p_\theta}\p_\theta,\nu\ra=-\frac{(\Sigma^2)_r}{2R}.
 \eee
 Hence
 \be
 \begin{split}
\kappa=&\frac{H_r}{2HR}+\frac{(\Sigma^2)_r}{2\Sigma^2R}=\frac{\Sigma H_r+2H\Sigma_r}{2HR\Sigma}=\frac{A}{2HR\Sigma^3}.
\end{split}
\ee
Hence
\be
 \begin{split}
 \int_{\mathcal{S}(r_0)}\kappa d\sigma=&\int_0^{2\pi}\lf(\int_0^\pi\frac{A}{2HR\Sigma^3} H^\frac12  \Sigma d\theta\ri)d\varphi\\
 =& \int_0^{2\pi}\lf(\int_0^\pi\frac{A}{2H^\frac12R\Sigma^2}   d\theta\ri)d\varphi
 \end{split}
\ee
On the other hand, at $x=0, y=0$,
\bee
\begin{split}
\int_0^\pi \frac{A\a}{D\ell}d\theta=&\int_0^\pi \frac{A R\Sigma}
{16\pi H^\frac12  R^2\Sigma\cdot \Sigma^2}d\theta\\
=&\int_0^\pi \frac{A  }
{16\pi R H^\frac12     \Sigma^2}d\theta\\
\end{split}
\eee
From the above two relations, the result follows.
\end{proof}

By direct computations, we have:
\begin{lma}\label{l-Gauss-1} The Gaussian curvature $K$ of $\mathcal{S}(r_0)$ is given by
\bee
K=\frac14\cdot \frac{- H \lf(2H_{\theta\theta}\Sigma^2- H_\theta(\Sigma^2)_\theta-4\Sigma^4\ri)-(4H\Sigma^2-H_\theta^2)\Sigma^2}{\Sigma^4H^2}.
\eee

\end{lma}

We want to isometrically embed the surface $(\mathcal{S}(r_0), d\sigma^2)$ in $\R^3$.  Consider the plane curve in the $xz$-plane, $(\eta(\theta),0,\xi(\theta))$, with
\be\label{e-embed-1}
\eta(\theta)=H^\frac12
\ee
\be\label{e-embed-2}
\xi(\theta)=\int_{\frac\pi2}^\theta\lf(\frac{ -H_\theta^2+4H\Sigma^2 }{4H}\ri)^\frac12d\theta.
\ee
Again we are fixing $t=t_0$, $r=r_0$.  Assume that $\eta(0)=\eta(\pi)=0$ and
assume that when we rotate the curve about $z$-axis we have a smooth embedded closed surface in $\R^3$. Then the surface is given by
\be\label{e-embed-3}
X(\theta,\varphi)=(x(\theta,\varphi),y(\theta,\varphi),z(\theta,\varphi))
=(\eta(\theta)\cos\varphi,\eta(\theta)\sin\varphi,\xi(\theta)).
\ee
It is easy to see that

$$
\la X_\theta,X_\theta\ra= \Sigma^2, \la X_\theta,X_\varphi\ra=0, \la X_\varphi,X_\varphi\ra=H.
$$
Hence $X$ is an isometric embedding of $(\mathcal{S}(r_0), d\sigma^2)$ in $\R^3$.

\begin{lma}\label{l-BY-2} Under the above assumptions and notations,   we have

\bee
\frac1{8\pi}\int_{\mathcal{S}(r_0)}\kappa_0 d\sigma=-\int_0^{2\pi}\lf(\int_0^\pi \frac CD\frac1\beta d\theta\ri)d\varphi
\eee
where $\kappa_0$ is the mean curvature with respect to the unit outward normal  of $\mathcal{S}(r_0)$ when it is isometrically embedded in $\R^3$, and $\beta$ is evaluated at $y=0$.
\end{lma}
\begin{proof} Let $s$ be the arc length of the curve $(\eta,0,\xi)$. Since the embedded surface is a surface of revolution, one of the eigenvalues $\lambda_1$ of the second fundamental form is \cite[section 3.3]{do Carmo}:
\bee
\begin{split}
\lambda_1=& \frac{d\xi}{ds}\frac1{\eta}\\
=&\frac{\xi'}{H^\frac12}(\frac{ds}{d\theta})^{-1}\\
=&\frac{1}{2H\Sigma}\lf(-H_\theta^2+4H\Sigma^2 \ri)^\frac12.
\end{split}
\eee

Using the expression for the Gaussian curvature $K$ in Lemma \ref{l-Gauss-1},   the mean curvature in $\R^3$ is
\bee
\begin{split}
\kappa_0=&\lambda_1+\lambda_2\\
=&\lambda_1+\frac K{\lambda_1}\\
=&\frac1{\lambda_1}(\lambda_1^2+K)\\
=&\frac1{\lambda_1}\lf[ \frac{ -H_\theta^2+4H\Sigma^2 }{4H^2\Sigma^2} +\frac14\cdot \frac{- H \lf(2H_{\theta\theta}\Sigma^2- H_\theta(\Sigma^2)_\theta-4\Sigma^4\ri)-(4H\Sigma^2-H_\theta^2)\Sigma^2}{\Sigma^4H^2}\ri]\\
=&\frac{-  H \lf(2H_{\theta\theta}\Sigma^2- H_\theta(\Sigma^2)_\theta-4\Sigma^4\ri)
 }{4\cdot \frac{1}{2H\Sigma}\lf(-H_\theta^2+4H\Sigma^2 \ri)^\frac12 H^2 \Sigma^4}\\
 =&\frac{-   \lf(2H_{\theta\theta}\Sigma^2- H_\theta(\Sigma^2)_\theta-4\Sigma^4\ri)
 }{2\Sigma^3\lf(  -H_\theta^2+4H\Sigma^2  \ri)^\frac12  }
\end{split}
\eee
Hence
\bee
\begin{split}
\int_{\mathcal{S}(r_0}\kappa_0d\sigma=&\int_0^{2\pi}\lf(\int_0^\pi \frac{-   \lf(2H_{\theta\theta}\Sigma^2- H_\theta(\Sigma^2)_\theta-4\Sigma^4\ri)
 }{2\Sigma^3\lf(  -H_\theta^2+4H\Sigma^2  \ri)^\frac12  }H^\frac12\Sigma d\theta\ri)d\varphi\\
 =&\int_0^{2\pi}\lf( \int_0^\pi \frac{- H^\frac12  \lf(2H_{\theta\theta}\Sigma^2- H_\theta(\Sigma^2)_\theta-4\Sigma^4\ri)
 }{2\Sigma^2\lf(  -H_\theta^2+4H\Sigma^2  \ri)^\frac12  }  d\theta\ri)d\varphi
 \end{split}
\eee

At $y=0$,
\bee
\frac CD\frac1\beta=\frac{H^\frac12(2H_{\theta\theta}\Sigma^2-H_\theta(\Sigma^2)_\theta-4\Sigma^4)}{16\pi \Sigma^2(-H_\theta^2+4H\Sigma^2)^\frac12}.
\eee
From these two relations, we conclude that the lemma is true.
\end{proof}
\begin{thm}\label{t-BY} For the Kerr like metric \eqref{e-Kerr-like-1}, suppose  the embedding \eqref{e-embed-1} is defined for the surface  $S$ with $r=$constant, $t=$constant, we have
$$
E(\mathbf{N}, \Omega(r);0,0)=\mathfrak{m}_{\mathrm{BY}}(\mathcal{S}(r))
$$
where $\mathfrak{m}_{BY}(\mathcal{S}(r))$ is the Brown-York mass (with respect to the slice $t$=constant) of $\mathcal{S}(r)$.

\end{thm}

\section{isometric embedding for surfaces in the Kerr spacetime}\label{s-embedding}

In this section, we will concentrate on the Kerr metric. The Kerr like metric \eqref{e-Kerr-like-1} now becomes:
\be\label{e-Kerr-1}
\left\{
  \begin{array}{ll}
    F=&{\displaystyle-\frac{\Delta-a^2\sin^2\theta}{\Sigma^2}},\\
   G=& {\displaystyle-\frac{4mar\sin^2\theta}{\Sigma^2}} \\
H=&{\displaystyle\frac{\sin^2\theta\lf((r^2+a^2)^2-\Delta a^2\sin^2\theta\ri)}{\Sigma^2}},\\
R^2=&{\displaystyle\frac{\Sigma^2}{\Delta}},\\
\Sigma^2=&r^2+a^2\cos^2\theta\\
\Delta=&  r^2-2mr+a^2.
  \end{array}
\right.
\ee
Our aim is to discuss where the surface $r=$cosntant, $t=$constant has positive Gaussian curvature and can be isometrically embedded in $\R^3$ as a surface of revolution. We always fix $t=t_0$ and denote the surface $r=$constant on the slice $t=t_0$ by $\mathcal{S}(r)$ as before.

In the rest of this section, we assume that $m=1$, $0\le a\le m$ and $r\ge r_+$ which is the largest root of $\triangle=0.$ We also use the following notations:

\be\label{e-notations}
\left\{
  \begin{array}{ll}
    s:=&\sin\theta; \\
    c:=&\cos\theta;\\
\e:=&\displaystyle{\frac ar};\\
p:=&\displaystyle{\frac{2\e^2s^2}{r(1+\e^2 c^2)}}\\
 q:=&\displaystyle{\frac{ 1+\e^2 }{(1+\e^2 c^2)}}.
  \end{array}
\right.
\ee

 Then
\be\label{e-shorthand}
\left\{
  \begin{array}{ll}
    \Sigma^2=&r^2(1+\e^2c^2) \\
\triangle=&r^2\displaystyle{(1+\e^2-\frac2r)}\\
    H=  &  r^2s^2(1+\e^2 +p) \\
  \end{array}
\right.
\ee

\begin{lma}\label{l-dH-ddH} We have
 \begin{enumerate}
\item [(i)]
$
(\Sigma^2)_\theta=-2r^2\e^2cs;
$
\item [(ii)] $
\displaystyle{\frac1{2r^2}H_\theta= cs\lf[(1+\e^2)+p(1+q)\ri]};
$
and
\item [(iii)]
\bee
\frac1{2r^2}H_{\theta\theta}=(c^2-s^2)\lf[(1+\e^2)+p(1+q)\ri]+2c^2\lf[  pq(1+q)+\frac{s^2\e^2}{(1+\e^2)}pq^2\ri].
\eee
 \end{enumerate}
\end{lma}
\begin{proof}
Note that
\bee
\begin{split}
p_\theta=&\frac{2\e^2}{r}\lf(\frac{2cs}{1+\e^2c^2}+\frac{\e^2s^2\cdot2cs}{(1+\e^2c^2)^2}\ri)\\
=&\frac{4cs\e^2}{r}\cdot\frac{1+\e^2}{(1+\e^2c^2)^2}.
\end{split}
\eee
So
\bee
s^2p_\theta=2cspq;\ \ q_\theta=\frac{2cs\e^2}{(1+\e^2)}q^2.
\eee
Hence
\bee
\begin{split}
\frac1{r^2}H_\theta=&2cs\lf(1+\e^2+p\ri)+s^2p_\theta\\
=&2 cs\lf[(1+\e^2)+p(1+q)\ri].
\end{split}
\eee

Hence
\bee
\begin{split}
\frac1{2r^2}H_{\theta\theta}=&(c^2-s^2)\lf[(1+\e^2)+p(1+q)\ri]+cs\lf[p_\theta(1+q)+pq_\theta\ri]\\
=&(c^2-s^2)\lf[(1+\e^2)+p(1+q)\ri]+2c^2\lf[  pq(1+q)+\frac{s^2\e^2}{(1+\e^2)}pq^2\ri].
\end{split}
\eee
\end{proof}

From \eqref{e-embed-2}, it is easy to see that a necessary condition so that the spacelike surface $r=$constant, $t=$constant can be isometrically embedded in $\R^3$ is: $4H\Sigma^2-H_\theta^2>0$. To estimate this expression, we have the following: (In the rest of the paper, $E$ always denotes a quantity which is bounded in absolute value by a constant which is independent of $r, a, \theta$ provided $r\ge 1$. Its meaning may vary from line to line).
\begin{lma}\label{l-beta}
$$
4H\Sigma^2-H_\theta^2=4r^4s^2\bigg\{(1+\e^2)s^2+\frac{2\e^2 s^2}r-2c^2p(1+\e^2)(1+q)-c^2p^2(1+q)^2\bigg\}.
$$
Moreover,
\begin{enumerate}
  \item [(i)] If $r\ge 2$, then
$$
4H\Sigma^2-H_\theta^2\ge  r^4s^4\frac{15}{64}.
$$
  \item [(ii)] If $r\ge 1$, then
$$
4H\Sigma^2-H_\theta^2=4r^4s^4(1+E\e^2).
$$
\end{enumerate}

\end{lma}
\begin{proof}
\bee
\begin{split}
&4H\Sigma^2-H_\theta^2\\
=&4r^4s^2\lf[(1+\e^2)+p\ri](1+\e^2 c^2)-4r^4c^2s^2\lf[(1+\e^2)+p(1+q)\ri]^2\\
=&4r^4s^2\bigg\{\lf[(1+\e^2)(1+\e^2 c^2)+\frac{2\e^2 s^2}r\ri]-c^2(1+\e^2)^2-2c^2p(1+\e^2)(1+q)-c^2p^2(1+q)^2\bigg\}\\
=&4r^4s^2\bigg\{(1+\e^2)s^2+\frac{2\e^2 s^2}r-2c^2p(1+\e^2)(1+q)-c^2p^2(1+q)^2\bigg\}.
\end{split}
\eee
This proves the first part of the lemma. (ii) follows from this immediately.

To prove (i), suppose $r\ge 2$, then $\e \le \frac12$.

\bee
\begin{split}
(1+\e^2)s^2+&\frac{2\e^2 s^2}r-2c^2p(1+\e^2)\\
=&(1+\e^2)s^2+\frac{2\e^2 s^2}r-\frac{4c^2\e^2s^2(1+\e^2)}{r(1+\e^2c^2)}\\
=&\frac{s^2}{1+\e^2c^2}\lf[(1+\e^2)(1+\e^2c^2)+\frac{2\e^2}r\lf(   s^2-\e^2c^2- c^2\ri)\ri]\\
\ge &\frac{s^2}{1+\e^2c^2}\lf[(1+\e^2)(1+\e^2c^2)+\e^2\lf(-\e^2c^2- c^2\ri)\ri]\\
=& s^2q.
\end{split}
\eee
Since
$$
p\le \frac14 s^2;\ 1\le q\le \frac54,
$$
 we have

\bee
\begin{split}
(1+\e^2)s^2+&\frac{2\e^2 s^2}r-2c^2p(1+\e^2)(1+q)-c^2p^2(1+q)^2\\
\ge& s^2q-2c^2p(1+\e^2)q-c^2p^2(1+q)^2\\
   \ge&  s^2 q(1-2\cdot\frac14\cdot\frac 54)-\frac{81}{256}s^2  \\
\ge &\frac{15}{256}s^2.
\end{split}
\eee
From this the lemma follows.

\end{proof}
Using \eqref{e-embed-1}, \eqref{e-embed-2} :
 \begin{cor}\label{c-embed-1} Suppose $r\ge 2$, or $r\ge r_+$ and $\e$ is small enough, then $\mathcal{S}(r)$ can be isometrically embedded in $\R^3$ as a surface revolution.
 \end{cor}
Next we want to estimate the Gaussian curvature of $\mathcal{S}(r)$. We will use the expression in Lemma \ref{l-Gauss-1}. First, we have the following:
\begin{lma}\label{l-estimate-C}
\bee
\begin{split}
&\frac1{4r^4}\bigg\{2H_{\theta\theta}\Sigma^2-H_\theta(\Sigma^2)_\theta-4\Sigma^4\bigg\}\\
=&-2s^2-\e^2s^2-\e^2c^2s^2+p(1+q)\lf[3c^2+2\e^2c^2+\e^2c^4-s^2\ri]\\
&+   2\e^2c^2s^2pq.
\end{split}
\eee
\end{lma}

\begin{proof}
\bee
\begin{split}
&\frac1{4r^4}\bigg\{2H_{\theta\theta}\Sigma^2-H_\theta(\Sigma^2)_\theta-4\Sigma^4\bigg\}\\
=&\lf\{(c^2-s^2)\lf[(1+\e^2)+p(1+q)\ri]+2c^2\lf[  pq(1+q)+\frac{s^2\e^2}{(1+\e^2)}pq^2\ri]\ri\}(1+\e^2c^2)\\
&+\e^2c^2s^2\lf[(1+\e^2)+p(1+q)\ri]-(1+\e^2c^2)^2\\
=&\lf\{-2s^2-\e^2 s^2+(c^2-s^2) p(1+q) +2c^2\lf[  pq(1+q)+\frac{s^2\e^2}{(1+\e^2)}pq^2\ri]\ri\}(1+\e^2c^2) \\
&+\e^2c^2s^2\lf[(1+\e^2)+p(1+q)\ri]\\
=&\lf(-2s^2-\e^2 s^2\ri)(1+\e^2c^2)+p(1+q)(c^2-s^2)(1+\e^2c^2) +2c^2p (1+q)(1+\e^2)+2\e^2c^2s^2pq \\
&+\e^2c^2s^2 (1+\e^2)+\e^2c^2s^2p(1+q) \\
=&\lf(-2s^2-\e^2 s^2\ri)(1+\e^2c^2)+p(1+q)\lf\{(c^2-s^2)(1+\e^2c^2)+2c^2\lf(1+\e^2\ri)+\e^2c^2s^2\ri\}\\
&+ 2 \e^2c^2s^2pq+\e^2c^2s^2(1+\e^2)\\
=&-2s^2-\e^2s^2-\e^2c^2s^2+p(1+q)\lf[3c^2+2\e^2c^2+\e^2c^4-s^2\ri]\\
&+   2\e^2c^2s^2pq.
\end{split}
\eee
\end{proof}

\begin{thm}\label{t-Gauss} Suppose $r\ge 2$, or $r\ge r_+$ and $\e$ is small then the Gaussian curvature of $\mathcal{S}(r)$ is positive.
\end{thm}
\begin{proof} By Lemma \ref{l-Gauss-1}
$$
K=\frac14\cdot \frac{- H \lf(2H_{\theta\theta}\Sigma^2- H_\theta(\Sigma^2)_\theta-4\Sigma^4\ri)-(4H\Sigma^2-H_\theta^2)\Sigma^2}{\Sigma^4H^2}.
$$
By Lemmas \ref{l-dH-ddH}, \ref{l-beta} and \ref{l-estimate-C},
we have
\be\label{e-K-1}
\begin{split}
\wt K:=&\frac1{4r^6s^2}\bigg\{- H \lf(2H_{\theta\theta}\Sigma^2- H_\theta(\Sigma^2)_\theta-4\Sigma^4\ri)-(4H\Sigma^2-H_\theta^2)\Sigma^2\bigg\}\\
\ge& (1+\e^2+p)\bigg\{ {  2s^2}+\e^2s^2+{\e^2c^2s^2}-{  c^2p(1+q)\lf[3 +2\e^2 +\e^2c^2\ri]}\\
 &-  {   2\e^2c^2s^2pq}\bigg\}\\
&-(1+\e^2c^2) \bigg\{{(1+\e^2)s^2}+\frac{2\e^2 s^2}r-{  2c^2p(1+\e^2)(1+q)-c^2p^2(1+q)^2}\bigg\}
\end{split}
\ee
It is easy to see that if $\e\ll 1$ and if $r\ge 1$, then $K>0$.

Suppose $r\ge 2$,
then  $p\le \frac14, q\le\frac 54$. So
\be\label{e-K-2}
\begin{split}
{\e^2c^2s^2-2\e^2c^2s^2pq}=&\e^2c^2s^2(1-2pq)\\
\ge &\e^2c^2s^2(1-\frac58)\\
=&\frac38\e^2c^2s^2.
\end{split}
\ee
\be
\begin{split}
{ 2s^2(1+\e^2+p)-s^2(1+\e^2c^2)(1+\e^2)}=& s^2\lf(2+2\e^2+2p-1-\e^2-\e^2c^2-\e^4c^2\ri)\\
=&s^2\lf(1+\e^2s^2+2p-\e^4c^2\ri)\\
\ge&s^2(\frac{15}{16}+\e^2s^2+2p).
\end{split}
\ee
Now
\bee\label{e-K-3}
\begin{split}
&{  -c^2p(1+q)(1+\e^2+p)\lf(3 +2\e^2 +\e^2c^2\ri)+2c^2 p(1+q) (1+\e^2c^2)(1+\e^2)}\\
=& c^2p(1+q)\lf[(1+\e^2)(2+2\e^2c^2-3-2\e^2-\e^2c^2)-p\lf(3 +2\e^2 +\e^2c^2\ri)\ri]\\
=&c^2p(1+q)\lf[(1+\e^2)( \e^2c^2-1-2\e^2)-p\lf(3 +2\e^2 +\e^2c^2\ri)\ri]
\end{split}
\eee
Hence
\be \label{e-K-4}
\begin{split}
&-c^2p(1+q)(1+\e^2+p)\lf(3 +2\e^2 +\e^2c^2\ri)+2c^2 p(1+q) (1+\e^2c^2)(1+\e^2)\\
&+c^2p^2(1+q)^2(1+\e^2c^2)\\
=&c^2p(1+q) (1+\e^2)( \e^2c^2-1-2\e^2)+c^2p^2(1+q)\lf[(1+q)(1+\e^2c^2)-(3 +2\e^2 +\e^2c^2)\ri]\\
=&c^2p(1+q) (1+\e^2)(- \e^2s^2-1- \e^2)-c^2p^2(1+q)(1+\e^2) \\
=&c^2p(1+q)(1+\e^2)( -\e^2s^2-1- \e^2-p).\\
\end{split}
\ee
By \eqref{e-K-1}--\eqref{e-K-4}, using the facts that $r\ge 2$, $p\le \e^2s^2\le \frac14$,$\e^2\le \frac14$ and $q\le \frac54$, we have
\bee
\begin{split}
\wt K\ge&(1+\e^2+p)(\e^2s^2+\frac38\e^2c^2s^2)+(\frac{15}{16}+\e^2s^2+2p)s^2\\
&-(1+\e^2c^2)\e^2s^2- c^2p(1+q)(1+\e^2)(  \e^2s^2+1+ \e^2+p)\\
=& \frac38\e^2c^2s^2(1+\e^2)+\e^2s^2(1+\e^2)+(\frac{15}{16}+\e^2s^2)s^2+ s^2p(\e^2+\frac38 \e^2c^2+2)\\
&-(1+\e^2c^2)\e^2s^2-c^2p(1+q)(1+\e^2)^2-c^2p(1+q)(1+\e^2)(\e^2s^2+p)\\
\ge&\frac38\e^2c^2s^2(1+\e^2) + (\frac{15}{16}+\e^2s^2)s^2-c^2p(1+q)(1+\e^2)^2\\
&+p\e^2c^2s^2\lf(1+\frac38+8- \frac94\cdot\frac54\cdot 2 \ri)\\
\ge&\e^2c^2s^2\lf(\frac38+\frac{14}{16}\cdot\e^{-2}  -\frac94\cdot(\frac54)^2\ri)+\frac1{16}s^2\\
\ge&\e^2c^2s^2\lf(\frac38+\frac{14}{16}\cdot 4-\frac94\cdot(\frac54)^2\ri)+\frac1{16}s^2\\
\ge& \frac1{16}s^2,
\end{split}
\eee
because $\e^{-2}\ge 4$.
 From this it is easy to see that $K>0$.

\end{proof}
By Theorems \ref{t-BY}, \ref{t-Gauss}  and Corollary \ref{c-embed-1} we have the following:
\begin{cor}\label{c-embed-2} For $r\ge 2$, or $r\ge r_+$ and $\e\ll 1$, the spacelike surface $S: r={\text constant}$, $t={\text constant}$, the Brown-York quasilocal energy of $S$ is given by
$$
\mathfrak{m}_{BY}(\mathcal{S}(r))=E(\mathbf{N},\Omega(r);0,0).
$$

\end{cor}

 \section{critical solutions to \eqref{e-xy-system-1}}\label{s-solution}

 As described in the introduction, we obtain the QLE in $E(\mathbf{N},\Omega(r);x,y)$ by finding critical value of $x, y$. Under the embedding of the form \eqref{e-TXYZ}, for the Kerr metric, $x, y$ should satisfy \eqref{e-xy-system-1}.
 In this section, we want to discuss the problems of existence and uniqueness of this system for the Kerr metric. First of all, let us rewrite the system.
 Recall the system \eqref{e-xy-system-1}
 is:
\be\label{e-xy-system-2}
\left\{
  \begin{array}{ll}
    y_\theta=  &\displaystyle{ -\frac{(\Sigma^2H)_r}{2HR^2}x+\lf(\frac{\Sigma_\theta}{\Sigma}-\frac{H_\theta}{2H}\ri)y } \\
    x_\theta= & \displaystyle{\frac{R_\theta}{R}x+\lf(\frac{(\Sigma^2H)_r}{2H\Sigma^2}-\frac{\a\beta+xy H_\theta}{2H\ell}\ri)y}.
  \end{array}
\right.
\ee
We consider the Kerr metric with $0\le a\le m$. We will discuss the system for $r>2m$ because of the results in section \ref{s-embedding}.
As in that section, let $s=\sin\theta$, $c=\cos\theta$, $\e=a/r$, $p= (2m\e^2s^2)/(r(1+\e^2 c^2))$, $q=(1+\e^2)/(1+\e^2c^2)$. Let
\be\label{e-wxy-1}
\left\{
  \begin{array}{ll}
    \wt H= &r^{-2}H=s^2(1+\e^2+p) ; \\
    \wt \Sigma^2=  & r^{-2}\Sigma^2=1+\e^2c^2; \\
\wt \triangle=&\displaystyle{r^{-2}\triangle=1+\e^2-\frac{2m}r};\\
\wt x=&\wt\triangle^\frac12 x;\\
\wt y=&r^{-1}y;\\
    \wt\ell =  & r^{-2}\ell=\wt y^2+1+\e^2 c^2;\\
\wt \a=&r^{-1} \wt\triangle^\frac12\a=\lf[(1+\e^2c^2)(\wt x^2+\wt y^2+1+\e^2c^2)\ri]^\frac12;\\
\wt \beta=&r^{-2}\beta=\displaystyle{2s\lf[ (1+\e^2+p) \wt y^2+(1+\e^2+\frac{2m\e^2}r)s^2-c^2p(1+q)(2(1+\e^2)+p(1+q)) \ri]^\frac12}
  \end{array}
\right.
\ee
where we have used  Lemmas \ref{l-dH-ddH} and \ref{l-beta}. Note that $\wt\a, \wt\beta, \wt\ell $ depend also the functions $\wt x, \wt y$. Direct computations show that
\bee
\frac{(\Sigma^2H)_r}{2HR^2}=\displaystyle{rP\wt\triangle  }
\eee
where $P=(1+\frac12p+q)/(1+\e^2+p)$. Also,
\bee
\frac{(\Sigma^2H)_r}{2H\Sigma^2}=\frac1r P, \ \ \frac{R_\theta}R=-\frac{\e^2 cs}{1+\e^2 c^2}
\eee
Using Lemmas \ref{l-dH-ddH}
 one can  check that \eqref{e-xy-system-2} is equivalent to
 \be\label{e-xy-system-3}
\left\{
  \begin{array}{ll}
    \wt y_\theta=  &\displaystyle{ - \wt\triangle^\frac12 P \wt x-\lf(\frac {\e^2cs}{1+\e^2 c^2} +\frac cs \lf(1+\frac{pq}{1+\e^2+p}\ri) \ri)\wt y } \\
   \wt x_\theta= & \displaystyle{\wt \triangle^\frac12 P\wt y-\frac{\e^2 cs}{1+\e^2 c^2}\wt x- \frac{\wt \a\wt\beta+\wt x\wt y \wt H_\theta}{2\wt H\wt \ell} \wt y}.
  \end{array}
\right.
\ee
Since $c/s$ and $\wt H_\theta/H$ will become infinite when $\theta\to0$ or $\pi$, we cannot use the apply standard theory to discuss the system. However, we still have   the following uniqueness result.
\begin{thm}\label{t-uniqueness} Let $\wt x_i, \wt y_i$, $i=1, 2$ be two sets of solutions to \eqref{e-xy-system-3} in $[0,\pi]$ with bounded derivatives.  Then $\wt y_i(0)=\wt y_i(\pi)=0$, $i=1 ,2$. Moreover, if  $\wt x_1(0)=\wt x_2(0)$ (or $\wt x_1(\pi)=\wt x_2(\pi))$, then  $\wt x_1=\wt x_2, \wt y_1=\wt y_2$.
\end{thm}
\begin{proof} It is easy to see that $P$ and $\wt \Delta$ are bounded below away from 0 in $[0,\pi]$. By the first equation of the system, we conclude that
 $|\wt y_i| \le c_1 s$ for some constant $c_1$, $i= 1,2$. In particular, $\wt y_i(0)=\wt y_i(\pi)=0$.

 Suppose  $\wt x_1(0)=\wt x_2(0)$ we want to prove that:
 \be\label{e-uniqueness-main}
 \frac12[(\wt y_1-\wt y_2)^2+(\wt x_1-\wt x_2)^2]_\theta\le c_2[(\wt y_1-\wt y_2)^2+(\wt x_1-\wt x_2)^2]
 \ee
 for some constant $c_2$ in $[0,\pi/2]$. Suppose this is true, then one can conclude that $[(\wt y_1-\wt y_2)^2+(\wt x_1-\wt x_2)^2]=0$ on $[0,\pi/2]$. Since the system is well-behaved in $(0,\pi)$ one can apply the uniqueness of solutions of ODE to conclude that the proposition is true for this case. The case that $\wt x_1(\pi)=\wt x_2(\pi)$ is similar.

  To prove \eqref{e-uniqueness-main},
by the first equation of the system, in $[0,\pi/2]$, we have
\be\label{e-uniqueness-1}
\begin{split}
&\frac12[(\wt y_1-\wt y_2)^2]_\theta\\
=&- \wt\triangle^\frac12 P (\wt x_1-\wt x_2)(\wt y_1-\wt y_2)-\lf(\frac {\e^2cs}{1+\e^2 c^2} +\frac cs \lf(1+\frac{pq}{1+\e^2+p}\ri) \ri)(\wt y_1-\wt y_2)^2\\
\le &c_3\lf[ (\wt x_1-\wt x_2)^2+(\wt y_1-\wt y_2)^2\ri]
\end{split}
\ee
for some constant $c_3>0$ because $c=\cos\theta\ge0$ in $[0,\pi/2]$.
To estimate $\frac12[(\wt x_1-\wt x_2)^2]_\theta$, let us denote $\wt \a(\wt x_i,\wt y_i,\theta)$ by $\wt\a_i$, and define $\wt \beta_i$, $\wt \ell_i$ similarly, $i=1, 2$.

One can check that
\be\label{e-uniqueness-2}
\begin{split}
|  \wt \a_1-\wt\a_2|=&\frac{\lf|\wt\a_1^2-\wt\a_2^2\ri|}{\wt\a_1+\wt\a_2}\\
=&\frac{\lf|(1+\e^2c^2)\lf(\wt x_1^2+\wt y_1^2-\wt x_2^2-\wt y_2^2\ri)\ri|}{\wt\a_1+\wt\a_2}\\
\le &  c_4(|\wt x_1-\wt x_2|+|\wt y_1-\wt y_2|)
\end{split}
\ee
for some constant $c_4$ which may also depend on the bounds of the solutions $\wt x_i,\wt y_i$, where we have used the fact that $\wt\a_i\ge 1$. Similarly,
\be\label{e-uniqueness-2-1}
\begin{split}
\lf|  \frac1{\wt \ell_1}-\frac1{\wt\ell_2}\ri|=&\frac{\lf|\wt y_1^2-\wt y_2^2\ri|}{\wt\ell_1\wt\ell_2}\\
\le &  c_5(|\wt y_1-\wt y_2|)
\end{split}
\ee
for some constant $c_5$ because $\wt \ell_i\ge 1$. By Lemma \ref{l-beta}, we have
$$
\wt \beta_i=r^{-2}\beta_i\ge \lf(\frac{15}{64}\ri)^\frac12s^2
$$
for some constant $c_5$ for $r>2m$. So we have
\be\label{e-uniqueness-3}
\begin{split}
|  \wt \beta_1-\wt\beta_2|=&\frac{\lf|\wt\beta_1^2-\wt\beta_2^2\ri|}{\wt\beta_1+\wt\beta_2}\\
=&\frac{ 4s^2(1+\e^2c^2)\lf|   \wt y_1^2 -\wt y_2^2 \ri|}{\wt\beta_1+\wt\beta_2}\\
=&\frac{ 4s^2(1+\e^2c^2)\lf|   (\wt y_1+\wt y_2)(\wt y_1-\wt y_2) \ri|}{\wt\beta_1+\wt\beta_2}\\
\le &  c_6s | \wt y_1-\wt y_2 |
\end{split}
\ee
for some constant $c_6>0$, where we have used the fact that near $\theta=0, \pi$, $|\wt y_i|\le c_{1}s$ for some constant $c_{1}$. By the second equation of \eqref{e-xy-system-3}, and by  \eqref{e-uniqueness-2}--\eqref{e-uniqueness-3}, we have
\bee
\begin{split}
&\frac12[(\wt x_1-\wt x_2)^2]_\theta\\
&\le c_7[(\wt y_1-\wt y_2)^2+(\wt x_1-\wt x_2)^2]+(\wt x_1-\wt x_2) \lf(\displaystyle{-\frac{\wt \a_1\wt\beta_1+\wt x_1\wt y_1 \wt H_\theta}{2\wt H\wt \ell_1} \wt y_1+ \frac{\wt \a_2\wt\beta_2+\wt x_2\wt y_2\wt H_\theta}{2\wt H\wt \ell_2} \wt y_2}\ri)\\
=&c_7[(\wt y_1-\wt y_2)^2+(\wt x_1-\wt x_2)^2]- \displaystyle{\frac{\wt \a_1\wt\beta_1+\wt x_1\wt y_1 \wt H_\theta}{2\wt H\wt \ell_1} (\wt y_1-\wt y_2)}(\wt x_1-\wt x_2) \\
&+ \displaystyle{\lf(\frac{\wt \a_2\wt\beta_2+\wt x_2\wt y_2\wt H_\theta}{2\wt H\wt \ell_2}-\frac{\wt \a_1\wt\beta_1+\wt x_1\wt y_1 \wt H_\theta}{2\wt H\wt \ell_1}\ri) (\wt x_1-\wt x_2) \wt y_2 }\\
\le &c_8[(\wt y_1-\wt y_2)^2+(\wt x_1-\wt x_2)^2]
\end{split}
\eee
for some constants $c_7, c_8$, where we have used the facts that $|\wt y_i|\le c_1s$, $|\wt \beta_i|\le c_9s^2$ for some constant $c_9$. Combining the above inequality with \eqref{e-uniqueness-1}, we conclude that \eqref{e-uniqueness-main} is true. This completes the proof of the theorem.

\end{proof}

We apply the theorem to the Minkowski spacetime. In this case, $a=m=0$ and \eqref{e-xy-system-3} becomes:

 \be\label{e-xy-system-M}
\left\{
  \begin{array}{ll}
    \wt y_\theta=  &\displaystyle{ - 2\wt x -\frac cs \wt y } \\
   \wt x_\theta= & \displaystyle{2 \wt y-\lf( \frac{(\wt x^2+\wt y^2+1)^\frac12(s^2+\wt y^2)^\frac12}{s(1+\wt y^2)}+\frac{\wt x\wt y c}{s(1+\wt y^2)}\ri)\wt y }.
  \end{array}
\right.
\ee
For this system, nontrivial solution exists: $\wt x=-k\cos\theta, \wt y=k\sin\theta$, $k=$constant, (see  \cite{SCLN:2013,JLTH}). So $x=-k\cos\theta, y=  kr\sin\theta$ solve \eqref{e-xy-system-2}.  The nontrivial solution corresponds to the \emph{inertial observer}. One can look at the displacement vector (see \cite{SCLN:2013} (58)). For the trivial solution ($k=0$), $\mathbf{N}=\partial_{t}$ is the \emph{static observer}; for constant $k$, the displacement vector is
\bee
 \mathbf{N}=\sqrt{1+k^2}\partial_{t}+k\cos\theta\partial_{r}-r^{-1}k\sin\theta\partial_{\theta}=\sqrt{1+k^2}\partial_{t}+k\partial_{x'}
\eee
for the coordinate transformation $x'=r\cos\theta$, $y'=r\sin\theta$. It is a Lorentz transformation in the $t-x'$ plane of the static observer with constant velocity $-k/\sqrt{1+k^2}$ in the $x'$ direction. One can check that the $E(\mathbf{N},\Omega(r);x,y)=0$ in this case. This reflects the fact that each inertial observer is equivalent and measures zero energy for Minkowski spacetime. By Theorem \ref{t-uniqueness}, these are the only solutions for the system \eqref{e-xy-system-2}.
\begin{cor}\label{c-Minkowski}
Let $x, y$ be solutions to the system \eqref{e-xy-system-2}. Suppose the derivatives of $x, y$ with respect to $\theta$ are bounded, then $x=-k\cos\theta$, $y=  kr\sin\theta$.
\end{cor}
\begin{proof} This follows from Theorem \ref{t-uniqueness} by letting $k=-x(0)$.

\end{proof}

 Next we want to prove that in case a horizon exists, then \eqref{e-xy-system-2} may only have trivial solutions for a fixed $r$. With the same notations as before. We still normalize so that $0\le a\le m=1$. Also, $r_+=1+(1-a^2)^\frac12$ is the larger root of $r^2-2r+a^2=0$.  We have:
 \begin{thm}\label{t-xy-system} In the Kerr spacetime, for $8/3>r>r_+$,  if $\e\ll 1$, then any solution to \eqref{e-xy-system-2} with $|x_\theta|, |y_\theta|$ being bounded must be the trivial  solution: $x\equiv0, y\equiv0$.

 \end{thm}
 Before we prove the theorem, we need the following:
 \begin{lma}\label{l-xy-system-1}
 For $8/3>r>r_+$,  if $\e\ll 1$,  then $\displaystyle{\frac{(\Sigma^2 H)_r}{2H\Sigma^2}-\frac{\a\beta}{2H\ell} < 0}$ in $(0,\pi)$ for any $x, y$.
 \end{lma}
 \begin{proof}
 If $\e$ is small, then
\bee
\begin{split}
\frac{\a\beta}\ell\ge&\frac{R\beta}{\ell^\frac12}\\
=&R\lf(-\frac{H_\theta^2}{\ell}+4H\ri)^\frac12\\
\ge& R \lf(-\frac{H_\theta^2}{\Sigma^2}+4H\ri)^\frac12\\
=&\frac R\Sigma \cdot 2r^2 s^2(1+E_1\e^2).\\
\end{split}
\eee
Here and below $E_i$ will denote a quantity which is bounded by  a constant independent of $r, \theta, a$ provided $r\ge 1$. Since,
\be\label{e-HSr}
(\Sigma^2H)_r=2r^3s^2\lf(2(1+\e^2)-\lf(1-\frac1r\ri)\e^2s^2\ri),
\ee
for $\frac83>r>r_+$,  if $\e\ll 1$, then we have
$$
 1+\frac{a^2}{r^2}-\frac2r<\frac14.
 $$
 Hence
\bee
\begin{split}
\frac{(\Sigma^2 H)_r}{2H\Sigma^2}-\frac{\a\beta}{2H\ell}\le &\frac1{2H\Sigma^2}\lf(4r^3s^2(1+E_2\e^2)-2r^2 s^2 \Sigma R(1+E_1\e^2)\ri)\\
=&\frac{2r^3s^2}{2H\Sigma^2}\lf(2(1+E_2\e^2)-\frac{1+E_3\e^2}{\lf(1+\frac{a^2}{r^2}-\frac2r\ri)^\frac12}\ri)\\
<0
\end{split}
\eee
for $\theta\in(0,\pi)$.
 This completes the proof of the lemma.
 \end{proof}
 \begin{proof}[Proof of Theorem \ref{t-xy-system}]
Suppose $x, y$ are solutions to the system in $[0,\pi]$ so that $x_\theta, y_\theta$ are bounded in $[0,\pi]$. Then   $|y|\le c_1s$ for some constant $c_1$ near $\theta=0, \pi$ by the proof of Theorem \ref{t-uniqueness}.
By Lemma \ref{l-xy-system-1}, for $8/3>r>r_+$,  if $\e\ll 1$, then
     we have:
\be\label{e-xy-system-xy}
\begin{split}
(xy)_\theta
=&\displaystyle{ -\frac{(\Sigma^2H)_r}{2HR^2}x^2-\frac{\Sigma H_\theta-2H\Sigma_\theta}{2H\Sigma}xy}+\displaystyle{\frac{R_\theta}{R}xy+\lf(\frac{(\Sigma^2H)_r}{2H\Sigma^2}-\frac{\a\beta+xy H_\theta}{2H\ell}\ri)y^2}\\
\le& \lf(\frac{R_\theta}R-\frac{\Sigma H_\theta-2H\Sigma_\theta}{2H\Sigma}  - H_\theta\cdot \frac{y^2}{2H\ell}\ri)xy\\
=& -P(r,\theta) xy-\frac{H_\theta }{2H}xy
\end{split}
\ee
where in the second line we have used Lemma \ref{l-xy-system-1} and the fact that $(\Sigma^2H)_r\ge0$ by \eqref{e-HSr}. Here  $$
-P=\frac{R_\theta}R  +\frac{\Sigma_\theta}\Sigma  -\frac{y^2H_\theta}{2H\ell}
$$
which is bounded on $[0,\pi]$. Let $Q=\int_{\frac\pi2}^\theta Pd\theta$. We need to be careful here because $H_\theta/H$ is not integrable. However, $H_\theta y/H$ is integrable because $|y|\le c_1 s$ near $\theta=0, \pi$. So $Q$ is continuous on $[0,\pi]$. Hence we have
\be\label{e-xy-system-xyeQ}
(xy\exp Q)_\theta\le -\frac{H_\theta}{2H}(xy\exp Q).
\ee
Let $\Psi=xy\exp Q$. Then the above inequality is:
\bee
\Psi_\theta\le -\frac{H_\theta}{2H}\Psi.
\eee
Suppose    $\Psi(\theta_0)\le -c_{10} $ for some $c_{10}>0$ for some $\pi>\theta_0> \pi/2$. Then $ H_\theta\Psi\ge0$ on $[\theta_0,\theta_0+\delta]\subset [\pi/2,\pi)$ for some $\delta>0$ because $H_\theta(\theta_0)<0$. Hence $\Psi$ is decreasing on this interval, which implies $\Psi\le-c_{10}$ in this interval. Continuing in this way, we conclude that $\Psi(\pi)\le -c_{10}$. This is impossible because $y(\pi)=0$ which implies $\Psi(\pi)=0$. Hence we conclude that $\Psi\ge 0$ on $[\pi/2,\pi]$.

Similarly, one can prove that $\Psi\le 0$ on $[0,\pi/2]$. In particular, $\Psi(\pi/2)=0$.

On the other hand, let
$$
W(\theta)=\int_{\frac\pi2}^\theta \frac{H_\theta}{2H}d\theta
$$
which is well defined on $(0,\pi)$. Then we have
$$
(\Psi\exp W)_\theta\le 0
$$
on $(0,\pi)$. Since $\Psi\exp W=0$ at $\pi/2$, we have $\Psi\exp W\le 0$ on $[\pi/2,\pi)$. In particular, we have $\Psi\le 0$ on $[\pi/2,\pi)$. Since $\Psi\ge 0$ on $[\pi/2,\pi)$, we have $\Psi\equiv0$ on $[\pi/2,\pi)$. Similarly, one can prove that $\Psi\equiv 0$ on $(0,\pi/2]$. To summarize, we have $\Psi\equiv 0$. This implies that $xy\equiv0$.

Suppose $x$ is never zero on $(0,\pi)$, then we must have $y\equiv0$. By the first equation of \eqref{e-xy-system-2}, we conclude that $x\equiv0$. This is a contradiction.
Hence $x(\theta_0)=0$ for some $\theta_0\in(0,\pi)$. Since $xy\equiv0$, by the second equation in \eqref{e-xy-system-2}, we have
$$
\frac12(x^2)_\theta=\frac{R_\theta}R x^2,
$$
which implies that $x^2\equiv0$. By  the second equation again, we have $y\equiv0$ because $xy\equiv0$ and by Lemma \ref{l-xy-system-1}, if $\theta\in(0,\pi)$
$$
\frac{(\Sigma^2 H)_r}{2H\Sigma^2}-\frac{\a\beta}{2H\ell}<0.
$$
 This completes the proof of the theorem.
\end{proof}

\section{Appendix: derivation of \eqref{e-QLE-Kerrlike}}
Since we have the boundary expression (see \cite{SCLN:2013} (48))
\be
\begin{split}\label{QLE B_CNT}
\mathcal{B}(\partial_{T})=&-\frac{\alpha(H\Sigma^2)_{r}}{2\kappa\sqrt{H}R^{2}\Sigma^2}\\
&-\frac{\sqrt{H}}{\kappa}\left(\frac{H_{\theta\theta}-2\ell}{\beta}+\frac{R_{\theta}xy}{R\alpha}
-\frac{xy^{3}\beta+H_{\theta}\alpha\Sigma^{2}}{\ell\alpha\beta\Sigma}\Sigma_{\theta}\right)\\
&+\frac{\sqrt{H}y}{\kappa\alpha}x_{\theta}+\frac{\sqrt{H}y(H_{\theta}\alpha-xy\beta)}{\kappa\ell\alpha\beta}y_{\theta},
\end{split}
\ee
we substitute \eqref{e-xy-system-1} into \eqref{QLE B_CNT} which becomes
\bee
\begin{split}
\mathcal{B}(\partial_{T})=&-\frac{\alpha(H\Sigma^2)_{r}}{2\kappa\sqrt{H}R^{2}\Sigma^2}\nonumber\\
&-\frac{\sqrt{H}}{\kappa}\left[
\frac{H_{\theta\theta}-2\ell}{\beta}+\frac{R_{\theta}xy}{R\alpha}
-\frac{xy^{3}\Sigma_{\theta}}{\ell\alpha\Sigma}-\frac{H_{\theta}\Sigma\Sigma_{\theta}}{\ell\beta}\right.\\
&-\displaystyle{\frac{R_\theta}{R\a}xy
-\lf(\frac{(\Sigma^2H)_ry^2}{2H\Sigma^2\a}-\frac{\beta y^2}{2H\ell}-\frac{xy^3 H_\theta}{2H\ell\a}\ri)}\\
&\left.\lf(-\frac{H_{\theta}y}{\ell\beta}+\frac{xy^2}{\ell\alpha}\ri)\left(\displaystyle{ -\frac{(\Sigma^2H)_r}{2HR^2}x-\frac{\Sigma H_\theta-2H\Sigma_\theta}{2H\Sigma}y}\ri)\right]\\
=&-\frac{\alpha(H\Sigma^2)_{r}}{2\kappa\sqrt{H}R^{2}\Sigma^2}
-\frac{\sqrt{H}}{\kappa}\left[
\frac{H_{\theta\theta}-2\ell}{\beta}
-\frac{H_{\theta}\Sigma\Sigma_{\theta}}{\ell\beta}\right.\\
&+\frac{\beta y^2}{2H\ell}
+\frac{H^2_\theta y^2}{2H\ell\beta}-\frac{H_{\theta}\Sigma_\theta y^2}{\ell\beta\Sigma}\\
&\lf.\displaystyle{-\frac{(\Sigma^2H)_ry^2}{2H\Sigma^2\a}+\frac{H_{\theta}}{\ell\beta}\frac{(\Sigma^2H)_r}{2HR^2}xy-\frac{x^2y^2}{\ell\alpha}\frac{(\Sigma^2H)_r}{2HR^2}}\right],
\end{split}
\eee
where the terms with $xy^3$ are canceled. By using $\beta^2=4H\ell-H^2_{\theta}$, the boundary term becomes
\bee
\begin{split}
\mathcal{B}(\partial_{T})=&-\frac{\alpha(H\Sigma^2)_{r}}{2\kappa\sqrt{H}R^{2}\Sigma^2}\nonumber\\
&-\frac{\sqrt{H}}{\kappa}\left[
\frac{H_{\theta\theta}-2\ell+2 y^2}{\beta}
-\frac{H_{\theta}\Sigma_{\theta}(\Sigma^2+y^2)}{\ell\beta\Sigma}\right.\\
&\lf.\displaystyle{-\frac{(\Sigma^2H)_ry^2}{2H\Sigma^2\a}+\frac{H_{\theta}}{\ell\beta}\frac{(\Sigma^2H)_r}{2HR^2}xy-\frac{x^2y^2}{\ell\alpha}\frac{(\Sigma^2H)_r}{2HR^2}}\right].
\end{split}
\eee
By using $\ell=y^2+\Sigma^2$ and $\a^2=x^2\Sigma^2+R^2\ell$, it is further simplied to be
\bee
\begin{split}
\mathcal{B}(\partial_{T})=&\frac{(H\Sigma^2)_{r}}{2\kappa\sqrt{H}R^{2}\Sigma^2}\lf(-\a+\frac{R^2y^2}{\a}+\frac{x^2\Sigma^2 y^2}{\ell\a}\ri)\\
&-\frac{\sqrt{H}}{\kappa}\left[
\frac{H_{\theta\theta}-2\Sigma^2}{\beta}
-\frac{H_{\theta}\Sigma_{\theta}}{\beta\Sigma}+\frac{H_{\theta}}{\ell\beta}\frac{(\Sigma^2H)_r}{2HR^2}xy\right]\\
=&-\frac{\Sigma(H\Sigma^2)_{r}}{2\kappa\sqrt{H}R^{2}\Sigma}\lf(\frac{\a}{\ell}\ri)\nonumber\\
&-\frac{1}{2\kappa\sqrt{H}R^2\Sigma}\left[
\frac{2HR^2(\Sigma H_{\theta\theta}-2\Sigma^3-H_{\theta}\Sigma_{\theta})}{\beta}
+\frac{\Sigma H_{\theta}(\Sigma^2H)_r}{\ell\beta}xy\right]\\
=&-\lf(\frac{A}{D}\frac{\a}{\ell}+\frac{C}{D}\frac{1}{\beta}+\frac{A H_{\theta}}{D}\frac{xy}{\beta\ell}\ri).
\end{split}
\eee

\end{document}